\documentclass[runningheads]{llncs}

\usepackage{microtype}
\usepackage{xcolor}
\usepackage{lineno}
\usepackage{hyperref}
\usepackage{placeins}
\usepackage{graphicx}              
\usepackage{amsmath}               
\usepackage{amsfonts}              
\usepackage{amssymb}
\usepackage{verbatim}
\usepackage{indentfirst}
\usepackage{enumerate}
\usepackage{url}
\usepackage{latexsym}
\usepackage{verbatim}
\usepackage{algorithm}
\usepackage{algorithmicx, algpseudocode}
\usepackage{color}
\usepackage{xcolor}
\usepackage{epsfig}
\usepackage{mathtools}
\usepackage{changepage}
\usepackage{morefloats}
\usepackage{caption}
\usepackage{array}
\usepackage{bm}
\usepackage{thmtools,thm-restate}
\usepackage{multirow}
\usepackage{hhline}		
\usepackage{diagbox}
\usepackage{xspace}
\allowdisplaybreaks

\begin{document}
\title{Inventory Routing Problem with Facility Location \thanks{This material is based upon research supported in part by the U. S. Office of Naval Research under award number N00014-18-1-2099, and the U. S. National Science Foundation under award number CCF-1527032.}}
%
%
\author{Yang Jiao\inst{1}\orcidID{0000-0001-9583-0784} \and
R. Ravi\inst{1}\orcidID{0000-0001-7603-1207}
}
\authorrunning{Y. Jiao and R. Ravi}
%
\institute{
Tepper School of Business, Carnegie Mellon University, 5000 Forbes Ave, Pittsburgh, PA 15213, USA\\
\email{\{yangjiao,ravi\}@andrew.cmu.edu}
}
\maketitle              
\begin{abstract}
We study problems that integrate depot location decisions along with the inventory routing problem of serving clients from these depots over time balancing the costs of routing vehicles from the depots with the holding costs of demand delivered before they are due. Since the inventory routing problem is already complex, we study the version that assumes that the daily vehicle routes are direct connections from the depot thus forming stars as solutions, and call this problem
	the Star Inventory Routing Problem with Facility Location (SIRPFL). As a stepping stone to solving SIRPFL, we first study the Inventory Access Problem (IAP), which is the single depot, single client special case of IRP. The Uncapacitated IAP is known to have a polynomial time dynamic program. We provide an NP-hardness reduction for Capacitated IAP where each demand cannot be split among different trips. We give a $3$-approximation for the case when demands can be split and a $6$-approximation for the unsplittable case. 
	For Uncapacitated SIRPFL, we provide a $12$-approximation by rounding an LP relaxation. Combining the ideas from Capacitated IAP and Uncapacitated SIRPFL, we obtain a $24$-approximation for Capacitated Splittable SIRPFL and a $48$-approximation for the most general version, the Capacitated Unsplittable SIRPFL.

\keywords{Inventory Routing Problem \and Facility Location \and Approximation Algorithms.}
\end{abstract}
\section{Introduction}
We initiate the integrated study of facility opening and inventory routing problems. Facility location has many applications such as the placement of factories, warehouses, service centers, etc. The facility location problem involves selecting a subset of locations to open facilities to serve demands minimizing the facility opening costs plus connection costs between demands and the opened locations. Inventory routing arises from Vendor Managed Inventory systems in which a product supplier and its retailers cooperate in the inventory planning. First, the retailers share with the supplier the demand patterns for its product and the storage costs for keeping early deliveries per retailer location. Then the supplier is responsible for planning a delivery schedule that serves all the demands on time. The inventory routing problem (IRP) trades off visits from fixed depots over a planning horizon to satisfy deterministic daily demands at clients to minimize routing costs plus the holding costs of demand delivered before they are due at clients. We integrate the decision of which depots to open in the problem and study the joint problem of opening depots (given opening costs), and using these depots to minimize the total inventory routing costs, i.e. the sum of the routing costs from these depots and the holding costs at clients.

The IRP has been challenging to study by itself from an approximation perspective: constant-factor approximations are known only in very structured metrics like trees~\cite{Cheungetal2016} or when the routes are periodic~\cite{FNR14}. Hence we simplify the routing considerably to gain a better understanding of the integrated problem. In particular, we assume that the visits from each client go to the closest opened depots via a direct edge: the routing solution is thus a collection of stars rather than the Steiner trees or tours considered in the original IRP. We call this simplified variant of IRP the Star IRP or SIRP for short.

\subsection{Problem Definitions}
The \emph{Star Inventory Routing Problem with Facility Location (SIRPFL)} is inventory routing with the extra choice to build depots at a subset of the locations for additional costs before the first day, which then can be used to route deliveries throughout the entire time horizon. Formally, we are given an undirected graph $G=(V,E)$ with edge weights $w_e$, a time horizon $1,\ldots,T$, a set $D$ of demand points $(v,t)$ with $d^v_t$ units of demand due by day $t$, facility opening costs $f_v$ for vertex $v$, holding costs $h^v_{s,t}$ per unit of demand delivered on day $s$ serving $(v,t)$. The objective is to open a set $F \in V$ of facilities that can be used throughout the entire time horizon, determine the set of demands to serve/visit per day, and connect any visited clients to opened facilities per day so that the total cost from facility openings, client-facility connections, and storage costs for early deliveries is minimized. Three natural variants of the problem arise based on whether the delivery vehicles are uncapacitated, and if not, whether or not any single day's demand can be split among different visits. We call the first variant the {\em Uncapacitated} version. For the {\em Capacitated} version, we assume all vehicles have a fixed capacity $U$ and arrive at two variants: {\em Unsplittable} where every daily demand is satisfied wholly in one visit and the {\em Splittable} where it can be split across multiple visits (even across multiple days). We assume that any single demand never exceeds the capacity of the vehicle so that the splittable problem is always feasible.

Once the facility decisions are fixed, the resulting SIRP instances can be decomposed across the clients due to the assumption that the visits are direct edges from the client to an open facility. Thus, for each client, the routing solution is the direct edge to the closest open facility, and the only decisions are the delivery days and in the capacitated case, the number of trips on such days. We call the single-depot single-client problem the {\em Inventory Access Problem (IAP)}. Even the simple IAP has the three variants alluded to above. 

\subsection{Contributions}

\begin{enumerate}
  \item We initiate the study of inventory routing problems integrated with Facility Location (IRPFL) and supply the first complexity and approximation results.
  \item For the simpler Inventory Access Problem, we show that the unsplittable capacitated case is already weakly NP-hard. The uncapacitated problem is a single-item lot-sizing problem, for which a polynomial time exact solution exists~\cite{WW58}. For the latter and its splittable counterpart, we give constant approximation algorithms using LP rounding.
  \item For the Star versions of the IRPFL we consider, we give constant-factor approximation for all three versions by deterministically rounding new linear programming relaxations for the problems. The table below summarizes the approximation guarantees.\\
  
  \begin{tabular}{ l | l | l }
     & IAP & SIRPFL \\ \hline
    Uncapacitated & polynomial time~\cite{WW58} & 12-approx \\ \hline
    Capacitated Splittable & 3-approx & 24-approx \\ \hline
     Capacitated Unsplittable & NP-hard, 6-approx & 48-approx 
  \end{tabular}

\vspace{.5cm}

  \item Our algorithms need to modify and adapt current facility location LP rounding methods, since none of the variables in the objective function can directly be used for the rounding methods. These methods may be useful in future work involving time-indexed formulations integrating network design and facility location.
\end{enumerate}

We review related work in the next section. We then present a complete description of our 12-approximation for the Uncapacitated SIRPFL in Section~\ref{sect:USIRPFL}. To handle the capacitated versions, we need to strengthen the LP relaxation to better bound multiple visits per day: we illustrate this using the simpler example of the Capacitated Splittable IAP by providing a 3-approximation in Section~\ref{sect:CIAP}. Building on the 3-approximation, we show a 6-approximation for the Capacitated Unsplittable IAP. Finally, Section~\ref{sect:conclusion} summarizes the results and open problems. Details of the other results are in the Appendix.


\section{Related Work}

{\bf UFL:} The first constant approximation for Uncapacitated Facility Location was a $3.16$-approximation by Shmoys et el.~\cite{STA97} using the filtering method of Lin and Vitter~\cite{LV92}. Various LP-based methods made further improvements~\cite{JV01,JMMSV03,CS04}. More recently, Li gave a $1.488$-approximation~\cite{Li2013}.

{\bf IRP:} \sloppy Without facility opening decisions, IRP itself on general metrics has an $O(\frac{\log T}{\log \log T})$-approximation by Nagarajan and Shi~\cite{NS16} and an $O(\log N)$-approximation by Fukunaga et al~\cite{FNR14}. For variants of periodic IRP, Fukunaga et al.~\cite{FNR14} provide constant approximations. IRP on tree metrics have a constant approximation~\cite{CELS16}. Another special case of IRP is the joint replenishment problem (JRP), which has also been extensively studied~\cite{AJR89,LRS06,LRSS08,NS09,BBCJNS14}.

{\bf TreeIRPFL:} Another related problem is the \emph{Tree IRPFL}, which has the same requirements except that the connected components for the daily visits are trees (instead of tours in the regular IRP, or stars in the Star IRP version we study). Tree IRPFL differs from Star IRPFL by allowing savings in connection costs by connecting clients through various other clients who are connected to an opened facility. 

Single-day variants of Tree IRPFL have been studied extensively. In these problems, there is no holding cost component and thus they trade off the facility location placements with the routing costs from these facilities. 
We use $\rho_{\Pi}$ to denote the best existing approximation ratio for problem $\Pi$.
For uncapacitated single-day Tree IRPFL, the problem can directly be modeled as a single Steiner tree problem: attach a new root node with edges to each facility of cost equal to its opening cost; finding a Steiner tree from this root to all the clients gives the required solution. Thus, this problem has a $\rho_{\text{ST}}$-approximation algorithm.
If clients are given in groups such that only one client per group needs to be served, Glicksman and Penn~\cite{GP08} generalize the Steiner tree approximation method of Goemans and Williamson~\cite{GW95} to $(2-\frac{1}{|V|-1})L$-approximation, where $L$ is the largest size of a group.
For the capacitated single-day case of Tree IRPFL, Harks et al.~\cite{HKM13} provide a 4.38-approximation. They also give constant approximations for the prize collecting variant and a cross-docking variant. For the group version of the problem, Harks and K\"{o}nig show a $4.38L$-approximation.

{\bf Integrated Logistics:} Ravi and Sinha~\cite{RS06} originated the study of more general integrated logistics problems, and give a $(\rho_{\text{ST}} + \rho_{\text{UFL}})$-approximation for a generalization of the capacitated single-day Tree IRPFL called Capacitated-Cable Facility Location (CCFL). Here, ST stands for Steiner Tree and UFL stands for Uncapacitated Facility Location. In CCFL, the amount of demand delivered through each edge must be supported by building enough copies of cables on the edge. They give a bicriteria $(\rho_{k-\text{MEDIAN}} + 2)$-approximation opening $2k$ depots for the $k$-median version of the CCFL, which allows $k$ depots to be located at no cost.

\section{Uncapacitated SIRPFL}
\label{sect:USIRPFL}
In this section, we give a constant approximation for Uncapacitated SIRPFL. First, we state the LP formulation for Uncapacitated SIRPFL. Let $z_v$ indicate whether a facility at $v$ is opened, $y^{uv}_s$ indicate whether edge $uv$ is built on day $s$, $y^{uv}_{st}$ indicate whether to deliver the demand of $(v,t)$ on day $s$ from facility $u$, and $x^v_{s,t}$ indicate whether demand point $(v,t)$ is served on day $s$. Then Uncapacitated SIRPFL has the following LP relaxation. To simplify notation, define $H^v_{s,t} = d^v_t h^v_{s,t}$, i.e., $H^v_{s,t}$ is holding cost of storing all of the demand for demand point $(v,t)$ from day $s$ to day $t$.

\begin{alignat}{2}
\min \qquad \sum_{v \in V} f_v z_v + &\sum_{s \leq T} \sum_{e \in E} w_e y^e_s + && \sum_{(v,t) \in D} \sum_{s \leq t} H^v_{s,t} x^v_{s,t} \nonumber \\
\text{s.t.} \hspace{1.4cm} \sum_{s \leq t} x^v_{s,t} &\geq 1 && \forall (v,t) \in D \label{constraint:service}\\
\sum_{u \in V} y^{uv}_{st} &\geq x^v_{s,t} && \forall (v,t) \in D, s \leq t \label{constraint:connection}\\
z_u &\geq \sum_{s=1}^T y^{uv}_{st} &&\forall (v,t) \in D, u \in V \label{constraint:facilityLB}\\
y^{uv}_s &\geq y^{uv}_{st} &&\forall (v,t) \in D, u \in V, s \leq t \label{constraint:connectionLB}\\
z_u \geq &y^{uv}_s &&\forall u,v \in V, s \leq T \label{constraint:facilityLB2}\\
\sum_{u \in V} \sum_{s = s'}^{t_2} y^{uv}_{s t_2} &\geq \sum_{u \in V} \sum_{s = s'}^{t_2} y^{uv}_{s t_1} &&\forall v \in V, t_2 > t_1 \geq s' \label{constraint:serviceIntervals}\\
z_u, y^e_r, y^a_{l,m} x^v_{s,t} &\geq 0 && \forall u, v \in V, e,a \in E, r, m, t \leq T, l \leq m, s \leq t \label{constraint:nonnegAll}.
\end{alignat}

Constraint~\ref{constraint:service} requires that every demand point is served by its deadline. Constraint~\ref{constraint:connection} enforces that $v$ gets connected to some facility on day $s$ if $(v,t)$ is served on day $s$. Constraint~\ref{constraint:facilityLB} ensures that facility $u$ is open if $u$ is assigned to any demand point over the time horizon. 
	Constraint~\ref{constraint:connectionLB} ensures that whenever $(v,t)$ is served on day $s$ from $u$, an edge between $u$ and $v$ must be built on day $s$. Constraint~\ref{constraint:facilityLB2} ensures that whenever some client $v$ is connected to $u$ on some day $s$, a facility must be built at $u$. Constraint~\ref{constraint:serviceIntervals} is valid for optimal solutions since for any $v$, if there is a service to $(v,t_1)$ within $[s',t_1]$ and $t_1 < t_2$, then the service to $t_2$ is either on the same day or later, i.e., there must be a service to $(v,t_2)$ within $[s',t_2]$. Here we are using the property that in an optimal solution the demands from a client over time are served in order without loss of generality, which is a consequence of the monotonicity of the unit holding costs at any location.

Using the above LP formulation, we provide an LP rounding algorithm. Before stating the algorithm, we define the necessary notation. First, let $(x,y,z)$ be an optimal LP solution. Let $f(x,y,z)$, $r(x,y,z)$, and $h(x,y,z)$ denote the facility cost, routing cost, and holding cost of $(x,y,z)$ respectively. Define $s_{v,t}$ to be the latest day $s^*$ such that $\sum_{u \in V} \sum_{s=s^*}^{t} y^{uv}_{st} \geq \frac{1}{2}$. 

The key idea is to apportion the visit variable $y_{uv}^s$ at day $s$ to different demand days $t$ that it serves using the additional variable $y_{uv}^{st}$. The latter variables for any demand at node $v$ on day $t$ provide a stronger lower bound, via Constraint~\ref{constraint:facilityLB}, on how much facility must be installed at node $u$ than any lower bound from $y^{uv}_s$ alone. Constraint~\ref{constraint:facilityLB} is a crucial component in the proof of Lemma~\ref{lem:concentration}, which ultimately allows us to bound the facility cost.

Ideally, we would like to use $s_{v,t}$ to bound the holding cost incurred when serving $(v,t)$ on day $s_{v,t}$. However, to avoid high routing costs, not all demands will get to be served by the desired $s_{v,t}$. Instead, for each client $v$, an appropriately chosen subset of $\{s_{v,t}: t \leq T\}$ will be selected to be the days that have service to $v$. To determine facility openings and client-facility connections, the idea is to pick ``balls" that gather enough density of $z_u$ values so that the cheapest facility within it can be paid for by the facility cost part of the LP objective. To be able to bound the routing cost, we would like to pick the radii of the balls based on the amount of $y^{uv}_s$ values available from the LP solution. However, $y^{uv}_s$ by itself does not give a good enough lower bound for $z_u$. So we will carefully assign disjoint portions of $y^{uv}_s$ to $y^{uv}_{st}$ for different $t$'s. In this way, we use $y^{uv}_{st}$ to bound the facility cost, and the disjoint portions of $y^{uv}_s$ to pay for the routing cost. With these goals in mind, we now formally define the visit days and the radius for each client.

Fix a client $v$. The set $A_v$ of demand days $t$ that $v$ gets visited on their $s_{v,t}$ will be assigned based on collecting enough $y^{uv}_{st}$ over $u$ and $s$. We call the days in $A_v$ \emph{anchors} of $v$. Denote by $t_{L_v}$ the latest day that has positive demand at $v$. We use $S_v$ to keep track of the service days for the anchors.

\begin{algorithm}
\caption{Visits for $v$}
\label{algor:visits}
\begin{algorithmic}[1]
\State Initialize $A_v \leftarrow \{t_{L_v}\}$.
\State Initialize $S_v \leftarrow \{s_{v,t_{L_v}}\}$.
\State Denote by $\tilde{t}$ the earliest anchor in $A_v$.
\While{there is a positive unserved demand at $v$ on some day before $\tilde{t}$}
	\State Denote by $t$ the latest day before $\tilde{t}$ with positive demand at $(v,t)$.
	\If{$t \geq s_{v,\tilde{t}}$}
		\State Serve $(v,t)$ on day $s_{v,\tilde{t}}$.	
	\Else
		\State Update $A_v \leftarrow A_v \cup \{t\}$.
		\State Update $S_v \leftarrow S_v \cup \{s_{v,t}\}$.
		\State Update $\tilde{t} \leftarrow t$.
	\EndIf				
\EndWhile
\State Output the visit set $S_v$ for $v$.
\end{algorithmic}
\end{algorithm}

Define $W_{v,t} = \sum_{u \in V} \sum_{s = s_{v,t}}^t w_{uv} y^{uv}_{st}$. Let $W_v = \min_{t \in A_v} W_{v,t}$. Finally, define $B_v = \{u \in V: w_{uv} \leq 4 W_v\}$, which is a ball of radius $4 W_v$ centered at $v$. For ball $B_v$, let $F_v = \arg\min_{q \in B_v} f_q$. Simply, $F_v$ is a location in $B_{v}$ with the lowest facility cost. Now we are ready to state the algorithm for opening facilities in Algorithm~\ref{algor:starsIRPFL}.

\begin{algorithm}
\caption{12-approximation for Uncapacitated SIRPFL}
\label{algor:starsIRPFL}
\begin{algorithmic}[1]
\State $\mathcal{B} \leftarrow \emptyset$
\While{there is any ball $B_v$ disjoint from all balls in $\mathcal{B}$}
\State Add to $\mathcal{B}$ the ball $B_{v_i}$ of smallest radius
\EndWhile
\State Within each ball $B_{v_i}$, open a facility at $F_{v_i}$.
\State Assign each client $v$ to the closest opened facility $u(v)$.
\State For each $v$, serve it on all days in $S_v$ by building an edge from facility $u(v)$ to $v$ per day $s \in S_v$.
\end{algorithmic}
\end{algorithm}

\noindent Denote by $B_{v_1},\ldots,B_{v_l}$ the balls picked into $\mathcal{B}$ by Algorithm~\ref{algor:starsIRPFL}.

\begin{proposition}
The holding cost of the solution from the algorithm is at most $2h(x,y,z)$.
\end{proposition}
\begin{proof}
For each demand point $(v,t)$, we will charge a disjoint part of twice the $x$ values in the LP solution to pay for the holding cost. In particular, to pay for the holding cost incurred by $(v,t)$, we charge $\sum_{s=1}^{s_{v,t}} H^v_{s,t} x^v_{s,t}$ part of the LP solution. We consider two cases: $t \in A_v$ and $t \notin A_v$.
\begin{enumerate}
\item In this case, assume that $t \in A_v$. Then $(v,t)$ is served on day $s_{v,t}$, and incurs a holding cost of $H^v_{s_{v,t},t}$. By definition of $s_{v,t}$, we have $\sum_{u \in V} \sum_{s=s_{v,t}+1}^t y^{uv}_{st} < \frac 1 2$. Then
$$\sum_{s=1}^{s_{v,t}} x^v_{s,t} \geq 1 - \sum_{s = s_{v,t} + 1}^t x^v_{s,t}
\geq 1 - \sum_{s = s_{v,t} + 1}^t \sum_{u \in V} y^{uv}_{st}
> 1 - \frac 1 2
= \frac 1 2.$$
So our budget of $\sum_{s=1}^{s_{v,t}} H^v_{s,t} x^v_{s,t}$ is at least $H^v_{s_{v,t},t} \sum_{s=1}^{s_{v,t}} x^v_{s,t} \geq \frac{ H^v_{s_{v,t},t} }{2}$.

\item In this case, assume that $t \notin A_v$. Let $\tilde{t}$ be the earliest anchor after $t$. Since $t$ is not an anchor, $[s_{v,t},t]$ must have overlapped $[s_{v,\tilde{t}},\tilde{t}]$. So $s_{v,\tilde{t}} \leq t$. So $(v,t)$ is served on $s_{v,\tilde{t}}$. By constraint~\ref{constraint:serviceIntervals}, we have $s_{v,t} \leq s_{v,\tilde{t}}$. By monotonicity of holding cost, the holding cost incurred by serving $(v,t)$ on $s_{v,\tilde{t}}$ is at most $H^v_{s_{v,t},t} \leq 2 \sum_{s=1}^{s_{v,t}} H^v_{s,t} x^v_{s,t}$.
\end{enumerate}

\end{proof}

\begin{proposition}
\label{prop:routingCostUSIRPFL}
The routing cost of the solution from the algorithm is at most $12r(x,y,z)$.
\end{proposition}
\begin{proof}
We will charge a disjoint portion of $12$ times the $y$ values in the LP solution to pay for the routing cost. Note that only anchors cause new visit days to be created in the algorithm. So consider a demand point $(v,t)$ such that $t$ is an anchor for $v$.
\begin{enumerate}
\item First, consider the case that $v \in \{v_1,\ldots,v_l\}$, the set of vertices for whose balls were picked in $\mathcal{B}$ in Algorithm~\ref{algor:starsIRPFL}. Then the routing cost to connect $(v,t)$ to the nearest opened facility is
\begin{align*}
w_{F_v,v} &\leq 4 W_v
\leq 4 W_{v,t} \text{(by definition of $W_v$)}\\
&\leq 4 \sum_{u \in V} \sum_{s = s_{v,t}}^t w_{uv} y^{uv}_s \text{ (by constraint~\ref{constraint:connectionLB})}.
\end{align*}
Since it is within 4 times the LP budget, the desired claim holds.

\item Now, assume that $v \notin \{v_1,\ldots,v_l\}$. Then $B_v$ overlaps $B_{v'}$ for some $v'$ of smaller radius than $B_v$ (otherwise $B_v$ would have been chosen into $\mathcal{B}$ instead of the larger balls that overlap $B_v$). Then the edge built to serve $(v,t)$ connects $F_{v'}$ to $v$. So the routing cost to serve $(v,t)$ is
\begin{align*}
w_{F_{v'},v} &\leq W_{v,v'} + W_{v',F_{v'}}\\
&\leq 2 \cdot 4 W_v + 4 W_v \text{ (since radius of $B_v$ is at least radius of $B_{v'}$)}\\
&\leq 12 W_v \leq 12 W_{v,t} \leq 12 \sum_{u \in V} \sum_{s = s_{v,t}}^t w_{uv} y^{uv}_s \text{ (by constraint~\ref{constraint:connectionLB})}.
\end{align*}
\end{enumerate}

Observe that for every $v$ and any two anchors $t_1, t_2$ for $v$, we have $[s_{v,t_1}, t_1] \cap [s_{v,t_2}, t_2] = \emptyset$ by the construction of anchors in Algorithm~\ref{algor:visits}. So each $y^{uv}_s$ is charged at most once among all demands whose deadline correspond to anchors.
\end{proof}

Before bounding the facility costs, we show a Lemma that will help prove the desired bound.

\begin{lemma}
\label{lem:concentration}
For all $i \in \{1,\ldots,l\}$, we have $\sum_{v \in B_{v_i}} z_v \geq \frac 1 4$.
\end{lemma}
\begin{proof}
Suppose there is some $i \in \{1,\ldots,l\}$ such that $\sum_{u \in B_{v_i}} z_u < \frac 1 4$. Let $\hat{t} = \arg\min_t W_{v_i,t}$. Then
\begin{align*}
W_{v_i} &= W_{v,\hat{t}}
=\sum_{u \in V} \sum_{s=s_{v,\hat{t}}}^{\hat{t}} w_{uv} y^{uv}_{s \hat{t}}
\geq \sum_{u \notin B_{v_i}} \sum_{s=s_{v,\hat{t}}}^{\hat{t}} w_{uv} y^{uv}_{s \hat{t}}\\
&\geq 4 W_v \sum_{u \notin B_{v_i}} \sum_{s=s_{v,\hat{t}}}^{\hat{t}} y^{uv}_{s \hat{t}} \text{ (since $u \notin B_{v_i}$)}\\
&\geq 4 W_v (\sum_{u \in V} \sum_{s=s_{v,\hat{t}}}^{\hat{t}} y^{uv}_{s \hat{t}} - \sum_{u \in B_{v_i}} \sum_{s=s_{v,\hat{t}}}^{\hat{t}} y^{uv}_{s \hat{t}})\\
&\geq 4 W_v (\frac 1 2 - \sum_{u \in B_{v_i}} \sum_{s=s_{v,\hat{t}}}^{\hat{t}} y^{uv}_{s \hat{t}}) \text{ (by definition of $s_{v,t}$)}\\
&\geq 4 W_v (\frac 1 2 - \sum_{u \in B_{v_i}} z_u) \text{ (by constraint~\ref{constraint:facilityLB})}\\
&> W_v \text{ (by the supposition $\sum_{u \in B_{v_i}} z_u < \frac 1 4$, which leads to a contradiction)}.
\end{align*}
\end{proof}

\begin{proposition}
The facility cost of the algorithm's solution is at most $4f(x,y,z)$.
\end{proposition}
\begin{proof}
We will charge four times the $z$ values of the LP solution to pay for the facilities opened by the algorithm. Since the balls picked by Algorithm~\ref{algor:starsIRPFL} are disjoint, we can pay for each facility opened using the LP value in its ball. Consider ball $B_{v_i}$ picked by the algorithm and its cheapest facility $F_{v_i}$. Then the cost of opening $F_{v_i}$ is at most $f_v$ for all $v \in B_{v_i}$. So the facility cost for $F_{v_i}$ is
$$f_{F_{v_i}} \leq 4 \sum_{v \in B_{v_i}} z_v f_{F_{v_i}} \leq 4 \sum_{v \in B_{v_i}} z_v f_v.$$
The first inequality follows from Lemma~\ref{lem:concentration}. The second inequality is due to $F_{v_i}$ being the cheapest facility in the ball.
\end{proof}

Since facility, holding and routing costs are bounded within $12$ times their respective optimal values, we have the following result.
\begin{theorem}
Algorithm~\ref{algor:starsIRPFL} is a 12-approximation for Uncapacitated SIRPFL.
\end{theorem}




\section{Capacitated IAP}
\label{sect:CIAP}
Recall that the \emph{Inventory Access Problem (IAP)} is the single client case of the Inventory Routing Problem. The only decision needed is to determine on each day whether to visit the client and how much supply to drop off. In SIRPFL, if we know where to build the facilities, then the best way to connect clients would be to the closest opened facility. So once facility openings are determined, the remaining problem decomposes into solving IAP for every client.

\subsection{A 3-approximation for Capacitated Splittable IAP}
\label{sect:CSIAP}
Here, we consider \emph{Capacitated Splittable IAP}, in which a single demand is allowed to be served in parts over multiple days. Let $W$ be the distance between the depot and the client. Denote by $h_{s,t}$ the holding cost to store one unit of demand from $s$ to deadline $t$. The demand with deadline $t$ is denoted by $d_t$. Recall that $U$ denotes the capacity of the vehicle. We model Capacitated Splittable IAP by the following LP relaxation. 
\begin{alignat}{2}
\min \qquad \sum_{s \leq T} W y_s + & \sum_{t \in D} \sum_{s \leq t} h_{s,t} d_t x_{s,t} &&\nonumber \\
\text{s.t.} \hspace{1.5cm} \sum_{s \leq t} x_{s,t} &\geq 1 && \forall t \in D \label{constraint:serviceIAP}\\
y_s &\geq \sum_{t = s}^T \frac{x_{s,t} d_t}{U} && \forall s \leq T \label{constraint:routeLBIAP}\\
y_s &\geq x_{s,t} && \forall t \leq T, s \leq t \label{constraint:routeLB2IAP}\\
x_{s,t} &\geq 0 &&\forall t \in D, s \leq t\\
y_s &\geq 0 &&\forall s \leq T
\end{alignat}
The variable $y_s$ indicates the number of trips on day $s$. Variable $x_{s,t}$ indicates the fraction of $d_t$ to deliver on day $s$. Note that the objective only counts the cost of the visit to the client as a single copy of the trip variable $y_s$ reflecting the star constraint (if a return trip needs to be accounted for, we can multiply this term by 2 and all our results generalize easily). Constraint~\ref{constraint:serviceIAP} requires that each demand becomes entirely delivered by the due date (possibly split over multiple days). Constraint~\ref{constraint:routeLBIAP} ensures that the total demand that day $s$ serves do not exceed the total capacity among all trips on day $s$. Constraint~\ref{constraint:routeLB2IAP} ensures that there is a trip whenever some delivery is made on day $s$.

Let $(x,y)$ be an optimal LP solution. For convenience of the analysis, let $r(x,y) = \sum_{s \leq T} W y_s$ and $h(x,y) = \sum_{t \in D} \sum_{s \leq t} h_{s,t} d_t x_{s,t}$ denote the routing and holding cost of the solution respectively. We will use the LP values $x_{s,t}$ to determine when to visit the client and which demands to drop per visit. For each $t \in D$, let $s_t$ be the latest day for which $\sum_{s=s_t}^t x_{s,t} \geq \frac 1 2$. We will keep track of a visit set $S$ of days when visits are scheduled along with an anchor set $A$ consisting of demand days that caused the creation of new visits.

\begin{algorithm}
\caption{Visit Rule for Capacitated Splittable IAP}
\label{algor:visitsIAP}
\begin{algorithmic}[1]
\State Initialize $A \leftarrow \emptyset$.
\State Initialize $S \leftarrow \emptyset$.
\While{there is any unsatisfied demand}
	\State Denote by $t$ the unsatisfied demand day with the latest $s_t$	
	\State $A \leftarrow A \cup \{t\}$.
	\State $S \leftarrow S \cup \{s_t\}$.
	\State Satisfy $t$ by dropping off $d_t$ on day $s_t$.
	\For{unsatisfied demand day $\hat{t} \geq s_{t}$}
		\State satisfy $\hat{t}$	by dropping off $d_{\hat{t}}$ on day $s_t$.
	\EndFor	
\EndWhile
\State Output the visit set $S$.
\end{algorithmic}
\end{algorithm}

For the analysis, denote by $T_s$ the set of all demand days $t$ such that $t$ was satisfied by $s$ in Algorithm~\ref{algor:visitsIAP}.

\begin{proposition}
The holding cost of the solution from Algorithm~\ref{algor:visitsIAP} is at most $2h(x,y)$.
\end{proposition}

\begin{proof}
\begin{enumerate}
\item Assume that $t \in A$. Then $t$ was served on day $s_t$, i.e., incurs holding cost $h_{s_t,t} d_t$. To pay for the holding cost, we use the following part of the LP.
$$\sum_{s=1}^{s_t} h_{s,t} d_t x_{s,t} \geq h_{s_t,t} d_t \sum_{s=1}^{s_t} x_{s,t}  \geq \frac{h_{s_t,t} d_t}{2}.$$

\item Assume that $t \notin A$. Let $\tilde{s}$ be the latest day in $S$ such that $\tilde{s} \leq t$. Then the holding cost incurred by the demand on day $t$ is $h_{\tilde{s},t} d_t$. By definition of the chosen visit days $S$, $t$ was not chosen as anchor because $s_t$ was earlier than $\tilde{s}$. So we pay for the holding cost using
\begin{align*}
\sum_{s=1}^{\tilde{s}} h_{s,t} d_t x_{s,t} &\geq \sum_{s=1}^{s_t} h_{s,t} d_t x_{s,t} \text{ (by $s_t \leq \tilde{s}$)}\\
&\geq \frac{h_{s_t,t} d_t} {2} \geq \frac{h_{\tilde{s},t} d_t}{2} \text{ (by monotonicity of holding costs)}.
\end{align*}
\end{enumerate}
\end{proof}

\begin{proposition}
\label{prop:routingCostCSIAP}
The routing cost of the solution from Algorithm~\ref{algor:visitsIAP} is at most $3r(x,y)$.
\end{proposition}

\begin{proof}
For each visit day $s_{\tilde{t}} \in S$, the number of trips made is $\left \lceil \frac{\sum_{t \in T_{s_{\tilde{t}}}} d_t}{U} \right \rceil \leq \frac{\sum_{t \in T_{s_{\tilde{t}}}} d_t}{U} + 1$. So the total number of trips made is at most $\sum_{\tilde{t} \in A} \left( \frac{\sum_{t \in T_{s_{\tilde{t}}}} d_t}{U} + 1 \right) \leq \left( \sum_{\tilde{t} \in A} \frac{\sum_{t \in T_{s_{\tilde{t}}}} d_t}{U} \right) + |A|$. We will use $3$ copies of $\sum_{s=1}^T y_s$ to pay for the routing cost--$1$ copy to pay for the first term and $2$ copies to pay for the second term. The total LP budget for the number of trips is
\begin{align*}
\sum_{s=1}^T y_s &\geq \frac{\sum_{s=1}^T \sum_{t=s}^T x_{s,t} d_t}{U} \text{ (by constraint~\ref{constraint:routeLBIAP})}\\
&\geq \frac{\sum_{t=1}^T \sum_{s=1}^t x_{s,t} d_t}{U} \geq \sum_{t=1}^T \frac{d_t}{U} \geq \sum_{\tilde{t} \in A} \frac{\sum_{t \in T_{s_{\tilde{t}}}} d_t}{U}.
\end{align*}

So we can pay for the first term using one copy of the LP budget from all the $y$ variables.

To pay for the second term, we will use constraint~\ref{constraint:routeLB2IAP} instead so that we can use disjoint intervals of $y$ for different anchors. In particular, for anchor $\tilde{t}$, we will charge
\begin{align*}
2 \sum_{s=s_{\tilde{t}}}^{\tilde{t}} y_s &\geq \sum_{s=s_{\tilde{t}}}^{\tilde{t}} x_{s,\tilde{t}} \geq 2 \cdot \frac 1 2 \text{ (by definition of $s_{\tilde{t}}$)}.
\end{align*}

By the construction of $A$, for any $t_1,t_2 \in A$, we have $[s_{t_1},t_1] \cap [s_{t_2},t_2] = \emptyset$. So the payment for different anchors use disjoint portions of $y$. Hence the second term can be paid for within $2$ copies of the budget provided by $y$.
\end{proof}

Since both holding and routing costs are bounded within $3$ times their respective optimal values, we have the following result.
\begin{theorem}
Algorithm~\ref{algor:visitsIAP} is a $3$-approximation for the Capacitated Splittable Inventory Access Problem.
\end{theorem}

\subsection{A 6-approximation for Capacitated Unsplittable IAP}
Here, we show that Capacitated Unsplittable IAP has a $2\alpha_{CSIAP}$-approximation, where $\alpha_{CSIAP}$ is the best approximation factor for Capacitated Splittable IAP.
\begin{proposition}
\label{prop:splitToUnsplit}
There is a $2\alpha_{CSIAP}$-approximation for Capacitated Unsplittable IAP.
\end{proposition}
\begin{proof}
Given a Capacitated Unsplittable IAP instance, solve the corresponding Capacitated Splittable IAP instance obtaining a solution $(x,y)$ with approximation factor $\alpha_{CSIAP}$. To obtain a solution that does not split the demands, we will repack the demands per visit day of $(x,y)$. For each visit day $s$ of the solution $(x,y)$, let $D^s$ be the set of demands assigned to be served on day $s$ by $(x,y)$. Let $D^s_{\leq 1/2} = \{t \in D^s: d_t \leq U/2\}$ and $D^s_{> 1/2} = D^s \setminus D^s_{\leq 1/2}$. Denote by $n(s)$ the number of trips on day $s$ in the splittable solution. Note that $n(s) \geq \lceil \frac{\sum_{t \in D^s} d_t}{U} \rceil$. 

For each trip, for each demand in $D^s_{> 1/2}$, give each demand its own trip. Then, fill all demands of $D^s_{\leq 1/2}$ (without splitting) greedily into the previous trips and new ones as long as the capacity is not exceeded. This means that all trips involving demands in $D^s_{\leq 1/2}$, except for possibly one trip, will be filled to strictly more than half the capacity. 
	Let $n'(s)$ be the number of trips in the unsplittable solution thus obtained. If there are no trips of more than half the capacity, then $n'(s) =1 = n(s)$.
	Otherwise, the total sum of demands across the trips is strictly more than $(n'(s) -1) \cdot \frac{U}{2}$. 
	Since $n(s)  \geq \lceil \frac{\sum_{t \in D^s} d_t}{U} \rceil$, we get $n(s) > \frac{n'(s) - 1}{2}$, i.e., $n'(s) < 2n(s) + 1$, which implies that $n'(s) \leq 2n(s)$ since $n'(s)$ is an integer.
         Since we kept all deliveries to the days they occurred in $(x,y)$, the holding cost does not change. Hence, the unsplittable solution has cost at most $2$ times the splittable solution.
\end{proof}

Applying Proposition~\ref{prop:splitToUnsplit} with the $2$-approximation for Capacitated Splittable IAP, we obtain the following result.
\begin{theorem}
Capacitated Unsplittable IAP has a $6$-approximation.
\end{theorem}

In the Appendix, we show weak NP-hardness for the Capacitated Unsplittable IAP.

\section{Conclusion}
\label{sect:conclusion}
We studied the Uncapacitated, Capacitated Unsplittable, and Capacitated Splittable variants of IAP and SIRPFL. For the Uncapacitated IAP, a polynomial time dynamic program is known~\cite{WW58}. For the Capacitated Splittable IAP, we proved a $3$-approximation by rounding the LP. For the Capacitated Unsplittable IAP, we gave an NP-hardness reduction from Number Partition and a $6$-approximation. For the more general Uncapacitated Star Inventory Routing Problem with Facility Location (Uncapacitated SIRPFL), we gave a $12$-approximation by combining rounding ideas from Facility Location and the visitation ideas from our $3$-approximation for Capacitated Splittable IAP. For Capacitated Splittable SIRPFL, we provided at $24$-approximation. Following that, we have a $48$-approximation for Capacitated Unsplittable SIRPFL.
It remains open whether Capacitated Splittable IAP is NP-hard.
Since we tried to keep the proofs simple and did not optimize for the approximation factors, it may not be difficult to improve the factors.


%
%
%
%
\bibliographystyle{splncs04}
\bibliography{database_19_04_27}
%

\newpage
\appendix

\section{Capacitated Unsplittable IAP}
\label{sect:CUIAP}
To model more realistic scenarios, we now impose a capacity $U$ on the supply vehicle. The vehicle may make multiple trips in one day to meet the required demands. Also, we assume that demands are \emph{unsplittable}, i.e., each demand is within capacity and must be completely delivered in one trip.

\subsection{Capacitated Unsplittable IAP is Weakly NP-hard}
We show that Capacitated Unsplittable IAP is weakly NP-hard by reducing from Number Partitioning.

\begin{definition}
In \emph{Number Partitioning}, we are given a set $S$ of positive integers and wish to determine whether there is a subset $X \subset S$ such that $\sum_{a \in X} a = \sum_{b \in S \setminus X} b$.
\end{definition}

\begin{theorem}
Number Partitioning $\leq_P$ Capacitated Unsplittable IAP.
\end{theorem}

\begin{proof}
Let $S$ be the set in a given instance of the Number Partitioning problem. We will create an instance $I$ of Capacitated Unsplittable IAP as follows. For each $a \in S$, create a demand point on day $1$ with $a$ units of demand. This means that serving all demands on the first day forms a valid optimal solution. (Note that we are not using the holding costs at all in this reduction.) Set the capacity of the vehicle to $U := \frac{\sum_{a \in S} a}{2}$. Let the distance between the depot and the client be any positive number $w_{rv}$. We will show that there is a solution of cost at most $2 w_{rv}$ to I if and only if $S$ has a valid number partitioning.

First, we prove the forward direction. Assume that there is a solution to $I$ of cost at most  $2 w_{rv}$. Since the total demand is $2U$, any solution must cost at least $2 w_{rv}$. Furthermore, the only way to obtain cost exactly $2 w_{rv}$ is to drop off exactly $U = \frac{\sum_{a \in S} a}{2}$ units of demand per trip in two trips total. Let $X$ be the set of numbers corresponding to the demands in the first trip. By definition of the trips, we have $\sum_{a \in X} a = U = \sum_{b \in S \setminus X} b$.

Second, we prove the backward direction. Assume that there is a number partitioning $X \subset S$. Then serving $X$ and $S \setminus X$ each in a trip on day $1$ forms a feasible solution since $\sum_{a \in X} a = U = \sum_{b \in S \setminus X} b$. The cost of each trip is $w_{rv}$, which yields a total cost of $2 w_{rv}$.
\end{proof}

\section{Capacitated Splittable SIRPFL}
\label{sect:CSSIRPFL}
In this section, we study SIRPFL with vehicle capacities. Formally, we have a vehicle starting at the depot with capacity $U$ that is allowed to make multiple trips per day, where the routing cost accounts for the multiplicity of trips. The satisfaction of each demand is allowed to be split among multiple trips. Using the same variables as the LP formulation for Uncapacitated SIRPFL, the Capacitated Splittable SIRPFL has the following LP relaxation.

\begin{alignat}{2}
\min \qquad \sum_{v \in V} f_v z_v + &\sum_{s \leq T} \sum_{e \in E} w_e y^e_s + && \sum_{(v,t) \in D} \sum_{s \leq t} H^v_{s,t} x^v_{s,t} \nonumber \\
\text{s.t.} \hspace{1.4cm} \sum_{s \leq t} x^v_{s,t} &\geq 1 && \forall (v,t) \in D \label{constraint:serviceCS}\\
\sum_{u \in V} y^{uv}_{st} &\geq x^v_{s,t} && \forall (v,t) \in D, s \leq t \label{constraint:connectionCS}\\
\sum_{u \in V} y^{uv}_s &\geq \sum_{t=s}^T \frac{d^v_t}{U} x^v_{s,t} && \forall v \in V, s \leq T \label{constraint:capConnectionCS}\\
z_u &\geq \sum_{s=1}^T y^{uv}_{st} &&\forall (v,t) \in D, u \in V \label{constraint:facilityLBCS}\\
y^{uv}_s &\geq y^{uv}_{st} &&\forall (v,t) \in D, u \in V, s \leq t \label{constraint:connectionLBCS}\\
y^{uv}_s &\geq \sum_{t = s}^T \frac{d^v_t}{U} y^{uv}_{st} && \forall u, v \in V, s \geq T \label{constraint:capConnectionsLBCS}\\
z_u \geq &y^{uv}_s &&\forall u,v \in V, s \leq T \label{constraint:facilityLB2CS}\\
z_u, y^e_r, y^a_{l,m} x^v_{s,t} &\geq 0 && \forall u, v \in V, e,a \in E, r, m, t \leq T, l \leq m, s \leq t \label{constraint:nonnegAllCS}.
\end{alignat}

Constraints~\ref{constraint:serviceCS},~\ref{constraint:connectionCS},~\ref{constraint:facilityLBCS},~\ref{constraint:connectionLBCS},~\ref{constraint:facilityLB2CS},~\ref{constraint:nonnegAll} are the same as in the LP for Uncapacitated SIRPFL. Constraint~\ref{constraint:capConnectionCS} requires that the number of trips to $v$ on day $s$ must be at least the total demand at $v$ that were served from day $s$ scaled by the capacity limit. Similarly, constraint~\ref{constraint:capConnectionsLBCS} requires that the number of trips from facility $u$ to client $v$ on day $s$ must be at least the total demand at $v$ served by $u$ from day $s$ scaled by the capacity limit.

Now, we round the LP. Let $(x,y,z)$ be an optimal LP solution and $f(x,y,z)$, $r(x,y,z)$, and $h(x,y,z)$ be the facility cost, routing cost, and holding cost of $(x,y,z)$, respectively. We use the same notation as Section~\ref{sect:CSIAP}. Let $s_{v,t}$ be the latest day $s^*$ such that $\sum_{v \in V} \sum_{s = s^*}^T y^{uv}_{st} \geq \frac 1 2$. For client $v$, we keep track of the set of demand days $t$ that $v$ will be visited exactly on their $s_{v,t}$ day. We use $S_v$ to keep track of all days of visits assigned for $v$. To determine the visit per client $v$, we apply the visit rule for Capacitated Splittable IAP to each $v$ independently.

\begin{algorithm}
\caption{Visits for $v$}
\label{algor:visitsCSSIRPFL}
\begin{algorithmic}[1]
\State Initialize $A_v \leftarrow \emptyset$.
\State Initialize $S_v \leftarrow \emptyset$.
\While{there is any unsatisfied demand}
	\State Denote by $t$ the unsatisfied demand day with the latest $s_{v,t}$	
	\State $A_v \leftarrow A_v \cup \{t\}$.
	\State $S_v \leftarrow S_v \cup \{s_{v,t}\}$.
	\State Satisfy $(v,t)$ by dropping off $d^v_t$ on day $s_{v,t}$.
	\For{unsatisfied demand day $\hat{t} \geq s_{v,t}$}
		\State satisfy $(v,\hat{t})$	by dropping off $d^v_{\hat{t}}$ on day $s_{v,t}$.
	\EndFor	
\EndWhile
\State Output the visit set $S_v$.
\end{algorithmic}
\end{algorithm}

We use $T^v_s$ to denote the set of demands $(v,t)$ who were assigned to be served on day $s$.

We keep the same definition of the balls as Section~\ref{sect:USIRPFL}, i.e., $W_{v,t} = \sum_{u \in V} \sum_{s = s_{v,t}}^t w_{uv} y^{uv}_{st}$; $W_v = \min_{t \in A_v} W_{v,t}$; $B_v = \{u \in V: w_{uv} \leq 4 W_v\}$; $F_v = \arg\min_{q \in B_v} f_q$. To determine the facility openings, we will apply the same procedure as Algorithm~\ref{algor:starsIRPFL}. For Capacitated Splittable SIRPFL, it will yield a $24$-approximation.

\begin{algorithm}
\caption{24-approximation for Capacitated Splittable SIRPFL}
\label{algor:CSSIRPFL}
\begin{algorithmic}[1]
\State $\mathcal{B} \leftarrow \emptyset$
\While{there is any ball $B_v$ disjoint from all balls in $\mathcal{B}$} 
\State Add to $\mathcal{B}$ the ball $B_{v_i}$ of smallest radius
\EndWhile
\State Within each ball $B_{v_i}$, open a facility at $F_{v_i}$.
\State Assign each client $v$ to the closest opened facility $u(v)$.
\State For each $v$, serve it on all days in $S_v$ by building multi-edges as needed to serve all demands assigned to each visit day $s \in S_v$ from facility $u(v)$ to $v$.
\end{algorithmic}
\end{algorithm}

\noindent Let $B_{v_1},\ldots,B_{v_l}$ be the balls picked into $\mathcal{B}$ by Algorithm~\ref{algor:CSSIRPFL}. Now we are ready to bound the costs of the solution from Algorithm~\ref{algor:CSSIRPFL}.

\begin{proposition}
The holding cost of the solution from Algorithm~\ref{algor:CSSIRPFL} is at most $2h(x,y,z)$.
\end{proposition}
\begin{proof}
We apply the same method as bounding the holding cost in Capacitated Splittable IAP. As before, we shall charge $\sum_{s=1}^{s_{v,t}} H^v_{s,t} x^v_{s,t}$ per demand point $(v,t)$ to pay for its holding cost. We consider the case $t \in A_v$ separately from $t \notin A_v$.
\begin{enumerate}
\item
Assume that $t \in A_v$. Then $(v,t)$ is satisfied on $s_{v,t}$, which incurs a holding cost of  $H^v_{s_{v,t},t}$. Since $s_{v,t}$ was the latest day that accumulates $\frac 1 2$ value from $s_{v,t}$ up to $t$, it must be that $\sum_{u \in V} \sum_{s=s_{v,t}+1}^t y^{uv}_{st} < \frac 1 2$. So

\begin{align*}
\sum_{s=1}^{s_{v,t}} x^v_{s,t} &\geq 1 - \sum_{s = s_{v,t} + 1}^t x^v_{s,t} \text{ by constraint~\ref{constraint:serviceCS}}\\
&\geq 1 - \sum_{s = s_{v,t} + 1}^t \sum_{u \in V} y^{uv}_{st} \text{ by constraint~\ref{constraint:connectionCS}}\\
&> 1 - \frac 1 2 \text{ by definition of $s_{v,t}$}\\
&= \frac 1 2
\end{align*}

Hence, $\sum_{s=1}^{s_{v,t}} H^v_{s,t} x^v_{s,t} \geq \frac{H^v_{s_{v,t},t}}{2}$ is able to pay for half the holding cost for $(v,t)$.

\item
Assume that $t \notin A_v$. Let $\tilde{s}$ be the latest day in $S_v$ such that $\tilde{s} \leq t$. Then $(v,t)$ incurs a holding cost of $H^v_{\tilde{s},t}$. Since $s_{v,t}$ was not chosen into $S_v$, it must be that $\tilde{s} \geq s_{v,t}$ and $\tilde{s}$ was chosen instead. Then we have the following lower bound on the LP budget: $\sum_{s=1}^{s_{v,t}} H^v_{s,t} x^v_{s,t} \geq \frac{H^v_{s_{v,t},t}}{2} \geq \frac{H^v_{\tilde{s},t}}{2}$. The first inequality follows from the analysis of the first case, and the second from the monotonicity of the holding cost.
\end{enumerate}
\end{proof}

\begin{proposition}
The routing cost of the solution from Algorithm~\ref{algor:CSSIRPFL} is at most $24 r(x,y,z)$.
\end{proposition}
\begin{proof}
We will combine the proof methods for bounding the routing cost in Uncapacitated SIRPFL and Capacitated Splittable IAP. Fix $v \in V$. By Proposition~\ref{prop:routingCostUSIRPFL}, the length of a trip from $v$ to its nearest opened facility is $4 W_v$ if $v \in \{v_1,\ldots,v_l\}$ and $12 W_v$ otherwise. For each visit day $s \in S_v$, the number of trips needed on that day is $\left \lceil \sum_{t \in T^v_{s}} \frac{d^v_t}{U} \right \rceil \leq \sum_{t \in T^v_{s}} \frac{d^v_t}{U} + 1$. Then the total number of trips to $v$ is at most $\sum_{\tilde{t} \in A_v} \sum_{t \in T^v_{s_{v,\tilde{t}}}} \left(\frac{d^v_t}{U} + 1 \right) = \sum_{\tilde{t} \in A_v} \sum_{t \in T^v_{s_{v,\tilde{t}}}} \frac{d^v_t}{U} + |A_v|$. So the total routing cost for $v$ is at most $12 \sum_{\tilde{t} \in A_v} \sum_{t \in T^v_{s_{v,\tilde{t}}}} \frac{d^v_t}{U} W_v + 12 |A_v| W_v$. As in the the analysis of Proposition~\ref{prop:routingCostCSIAP}, we pay for the first term separately from the second term.

Now, we show that $12$ copies of the LP budget $\sum_{u \in V} \sum_{s=1}^T w_{uv} y^{uv}_{s}$ for $v$ suffices to pay for the first term.


\begin{align*}
12 \sum_{\tilde{t} \in A_v} \sum_{t \in T^v_{s_{v,\tilde{t}}}} \frac{d^v_t}{U} W_v &\leq 12 \sum_{\tilde{t} \in A_v} \sum_{t \in T^v_{s_{v,\tilde{t}}}} \frac{d^v_t}{U} W_{v,t}\\
&= 12 \sum_{\tilde{t} \in A_v} \sum_{t \in T^v_{s_{v,\tilde{t}}}} \frac{d^v_t}{U} \sum_{u \in V} \sum_{s = s_{v,t}}^t w_{uv} y^{uv}_{st}\\
&\leq 12 \sum_{t = 1}^T \frac{d^v_t}{U} \sum_{u \in V} \sum_{s=s_{v,t}}^t w_{uv} y^{uv}_{st} \text{ since each $(v,t)$ is assigned to some anchor}\\
&= 12 \sum_{u \in V} w_{uv} \sum_{t = 1}^T \sum_{s=s_{v,t}}^t \frac{d^v_t}{U} y^{uv}_{st}\\
&\leq 12 \sum_{u \in V} w_{uv} \sum_{t = 1}^T \sum_{s=1}^t \frac{d^v_t}{U} y^{uv}_{st}\\
&= 12 \sum_{u \in V} w_{uv} \sum_{s=1}^T \sum_{t=s}^T \frac{d^v_t}{U} y^{uv}_{st}\\
&\leq 12 \sum_{u \in V} w_{uv} \sum_{s=1}^T y^{uv}_s \text{ by constraint~\ref{constraint:capConnectionsLBCS}}.
\end{align*}

Next, we will pay for the second term by charging $12 \sum_{u \in V} \sum_{s=s_{v,\tilde{t}}}^{\tilde{t}} w_{uv} y^{uv}_{s}$ per anchor $(v,\tilde{t})$. The cost for each anchor is
\begin{align*}
12 W_v &\leq 12 W_{v,\tilde{t}}\\
&= 12 \sum_{u \in V} w_{uv} \sum_{s=s_{v,\tilde{t}}}^{\tilde{t}} y^{uv}_{s \tilde{t}}\\
&\leq 12 \sum_{u \in V} w_{uv} \sum_{s=s_{v,\tilde{t}}}^{\tilde{t}} y^{uv}_{s} \text{ by constraint~\ref{constraint:connectionLBCS}}\\
&= 12 \sum_{u \in V} \sum_{s=s_{v,\tilde{t}}}^{\tilde{t}} w_{uv} y^{uv}_{s}.
\end{align*}

Since the intervals among the anchors in $A_v$ are disjoint, only one copy of $\sum_{u \in V} \sum_{s=1}^T w_{uv} y^{uv}_{s}$ is used by $\sum_{u \in V} \sum_{\tilde{t} \in A_v} \sum_{s={s_{v,\tilde{t}}}}^{\tilde{t}} y^{uv}_s$. In total, we used $24$ copies of the routing cost of the LP budget to pay for all visits.
\end{proof}

\begin{proposition}
The facility cost of the solution from Algorithm~\ref{algor:CSSIRPFL} is at most $4 f(x,y,z)$.
\end{proposition}
\begin{proof}
We apply the proof for bounding the facility cost in Uncapacitated SIRPFL. The details are provided here for completeness. As we saw in Section~\ref{sect:USIRPFL}, it suffices to show that $\sum_{v \in B_i} z_v \geq \frac 1 4$ for all $i \in \{1,\ldots,l\}$.

Suppose for contradiction that there is some $i \in \{1,\ldots,l\}$ such that $\sum_{v \in B_i} z_v < \frac 1 4$. Define $\hat{t} = \arg \min_t W_{v,t}$. Then

\begin{align*}
W_{v_i} &= W_{v,\hat{t}}\\
&=\sum_{u \in V} \sum_{s=s_{v,\hat{t}}}^{\hat{t}} w_{uv} y^{uv}_{s \hat{t}}\\ 
&\geq \sum_{u \notin B_{v_i}} \sum_{s=s_{v,\hat{t}}}^{\hat{t}} w_{uv} y^{uv}_{s \hat{t}}\\
&\geq 4 W_v \sum_{u \notin B_{v_i}} \sum_{s=s_{v,\hat{t}}}^{\hat{t}} y^{uv}_{s \hat{t}} \text{ by $u \notin B_{v_i}$}\\
&\geq 4 W_v (\sum_{u \in V} \sum_{s=s_{v,\hat{t}}}^{\hat{t}} y^{uv}_{s \hat{t}} - \sum_{u \in B_{v_i}} \sum_{s=s_{v,\hat{t}}}^{\hat{t}} y^{uv}_{s \hat{t}})\\
&\geq 4 W_v (\frac 1 2 - \sum_{u \in B_{v_i}} \sum_{s=s_{v,\hat{t}}}^{\hat{t}} y^{uv}_{s \hat{t}}) \text{ by definition of $s_{v,t}$}\\
&\geq 4 W_v (\frac 1 2 - \sum_{u \in B_{v_i}} z_u) \text{ by constraint~\ref{constraint:facilityLBCS}}\\
&> W_v \text{ by the supposition that $\sum_{u \in B_{v_i}} z_u < \frac 1 4$, which leads to a contradiction}.
\end{align*}

\end{proof}

Putting together all of the bounds, we obtain the following result.

\begin{theorem}
Algorithm~\ref{algor:CSSIRPFL} is a $24$-approximation for Capacitated Splittable SIRPFL.
\end{theorem}


\section{Capacitated Unsplittable SIRPFL}
\label{sect:CUSIRPFL}
Here, we assume that each demand $d^v_t$ does not exceed $U$ and that $d^v_t$ is \emph{unsplittable}, which means that $d^v_t$ must all be delivered in one trip. Applying the conversion of splittable to unsplittable solution in Proposition~\ref{prop:splitToUnsplit} for each $v \in V$ yields a $2\alpha_{CSSIRPFL}$-approximation for Capacitated Unsplittable SIRPFL, where $\alpha_{CSSIRPFL}$ is the best approximation factor for Capacitated Splittable SIRPFL.

\begin{corollary}
\label{cor:splitToUnsplit}
There is a $2\alpha_{CSSIRPFL}$-approximation for Capacitated Unsplittable SIRPFL.
\end{corollary}

Since Capacitated Splittable SIRPFL has a $24$-approximation, we obtain the following result applying Corollary~\ref{cor:splitToUnsplit}.

\begin{theorem}
Capacitated Unsplittable IRPFL has a $48$-approximation.
\end{theorem}

\end{document}